\documentclass[12pt]{article}
\usepackage[english]{babel} 
\usepackage[utf8]{inputenc} 
\usepackage{marvosym}
\usepackage{enumitem}
\usepackage[scaled]{helvet}
\usepackage{hyperref}
\usepackage{mathrsfs,amssymb}
\usepackage[intlimits]{empheq}
\mathchardef\Rightarrow="3229
\usepackage{stmaryrd}
\newfont{\suet}{suet14}
\newfont{\schwell}{schwell}
\DeclareTextFontCommand{\textsuet}{\suet}
\DeclareTextFontCommand{\textschwell}{\schwell}
\usepackage{xcolor}
\usepackage{framed}
\usepackage{dsfont}
\usepackage{amsthm}
\usepackage{amsfonts}
\usepackage{subfig}
\usepackage{graphicx}
\usepackage{epstopdf}

\theoremstyle{definition}
\newtheorem{Def}{Definition}

\theoremstyle{plain}
\newtheorem{Lem}[Def]{Lemma}

\theoremstyle{plain}
\newtheorem{The}[Def]{Theorem}

\theoremstyle{plain}

\theoremstyle{remark}
\newtheorem{Rem}[Def]{Remark}

\theoremstyle{definition}
\newtheorem*{def*}{Definition}

\numberwithin{equation}{section}
\DeclareMathOperator{\ran}{Ran}

\newcommand{\footnoteremember}[2]{\footnote{#2}\newcounter{#1}%
\setcounter{#1}{\value{footnote}}
}
\newcommand{\footnoterecall}[1]{%
\footnotemark[\value{#1}]
}

\begin{document}

\title{On the minimax principle for Coulomb--Dirac operators}
\author{Sergey Morozov\footnoteremember{1}{Mathematisches Institut, LMU M\"unchen, Theresienstr. 39, 80333 Munich, Germany}\footnote{E-mail: morozov@math.lmu.de} \and David M\"uller\footnoterecall{1}\footnote{E-mail: dmueller@math.lmu.de}}
\date{}

\maketitle

\begin{abstract}
Let $q$ and $v$ be symmetric sesquilinear forms such that $v$ is a form perturbation of $q$. Then we can associate a unique self-adjoint operator $B$ to $q+ v$. Assuming that $B$ has a gap $(a,b)\subset\mathbb{R}$ in the 
essential spectrum, we prove a minimax principle for the eigenvalues of $B$ in $(a, b)$ using a suitable orthogonal decomposition of the domain of $q$. This allows us to justify two minimax characterisations of eigenvalues in the gap of three--dimensional Dirac operators with electrostatic potentials having strong Coulomb singularities.
\end{abstract}

\section{Introduction and main results}

\subsection{General discussion}

Since the early days of quantum mechanics the Dirac operators with potentials having a Coulomb singularity are used to describe relativistic electrons in atomic fields. We say that a measurable Hermitian $4\times 4$--matrix function $V$ on $\mathbb R^3$ belongs to the class $\mathcal P_{\nu}$, if for some $\widetilde\nu\in [0, \nu)$ the inequalities
\begin{align}\label{ad_pot}
       0\geqslant V(x) \geqslant -\frac{\widetilde\nu}{|x|}\mathds{1}_{\mathbb C^4}\quad \text{hold for almost every}\ x \in \mathbb{R}^3.
\end{align}
If $V\in\mathcal P_1$ and \eqref{ad_pot} is satisfied with $\widetilde\nu= \nu$, we say that $V\in \overline{\mathcal P}_\nu$.

Let $H_0$ be the free Dirac operator (see Appendix). If $V\in\mathcal P_1$, one can define a physically meaningful self-adjoint operator $H$ formally corresponding to $H_{0}+V$, see Subsection~\ref{applications subsection} below. For the essential spectra we have (see \cite{Nen1976})
\[
 \sigma_{\text{ess}}(H)= \sigma_{\text{ess}}(H_0)= (-\infty, -1]\cup[1, \infty).
\]
The eigenvalues of $H$ in $(-1,1)$ are of particular interest; for example, the lowest eigenvalue $\lambda_1$ in this gap is interpreted as the ground state energy of the electron.

In the rest of this subsection we assume that $V$ is an electric potential, i.e., is proportional to $\mathds{1}_{\mathbb C^4}$.

Talman \cite{Talman} and Datta and Devaiah \cite{Datta} proposed a formal minimax characterisation of $\lambda_1$:
\begin{align*}
             \lambda_1 = \min_{x\in \ran T_{+}}\max_{y\in \ran T_{-}}\frac{\langle x+y,(H_0 +V)(x+y)\rangle}{\| x+y \|^2}.
\end{align*}
Here $T_\pm$ are the projectors on the upper and lower two components of 4--spinors, i.e.,
\begin{equation}\label{T_pm}
 T_+\binom{\varphi}{\psi}:= \binom{\varphi}{0}, \quad T_-\binom{\varphi}{\psi}:= \binom{0}{\psi}, \quad\text{for}\quad \varphi, \psi\in\mathsf L^2(\mathbb R^3, \mathbb C^2).
\end{equation}
Esteban and S{\'e}r{\'e} \cite{multiplicity} replaced $T_\pm$ by the spectral projectors of the unperturbed Dirac operator $H_0$
\begin{equation}\label{Ps}
 P_+:= P_{H_0}\big([1, \infty)\big), \quad P_-:= P_{H_0}\big((-\infty, -1]\big)
\end{equation}
and announced that for $V\in\mathcal P_{1/2}$ the $k^{\textrm{th}}$ eigenvalue in the gap (counted from below with multiplicity) coincides with the minimax level
\begin{align}
             \lambda_k = \inf_{\substack{\mathfrak V \textrm{{ subspace of }}P_+\mathsf H^{1/2}(\mathbb{R}^3, \mathbb{C}^4) \\ \dim\mathfrak V=k}}\sup\limits_{x\in(\mathfrak V\oplus 
              P_-\mathsf H^{1/2}(\mathbb{R}^3, \mathbb{C}^4))\setminus \{ 0 \}}\frac{h_0[x]+v[x]}{\|x\|^2},\label{lambda_k}
\end{align}
where $h_0$ and $v$ are the quadratic forms of $H_0$ and $V$, respectively.

A general result on the variational characterisation of the eigenvalues of operators with gaps in the essential spectrum was proved by Griesemer and Siedentop \cite{Min1999}. As a corollary they found that the variational characterisation of the lowest eigenvalue by Talman, Datta, and Devaiah is correct for $-2\mathds1_{\mathbb C^4}< V\leqslant 0$ provided $V(x)\to 0$ as $|x|\to \infty$. Griesemer, Lewis, and Siedentop \cite{Lewis} proved that the approach of \cite{multiplicity} holds for $V\in P_{\gamma}$ where $\gamma\approx 0.3$ is the real solution of $2\gamma^3-3\gamma^2+4\gamma=1$. Dolbeault, Esteban and S{\'e}r{\'e} \cite{Maria}  extended the result of \cite{multiplicity} to a class of $V$ which, under an extra assumption of slow decay at infinity, contains $\mathcal P_{2/(2/\pi+ \pi/2)}$.
In \cite{Gap2000}, the same authors have claimed the validity of both Esteban--S{\'e}r{\'e} and Talman--Datta--Devaiah minimax principles for $V\in \mathcal P_1$. However, they replaced $P_\pm\mathsf H^{1/2}(\mathbb{R}^3, \mathbb{C}^4)$ in \eqref{lambda_k} by $P_\pm C_0^\infty(\mathbb R^3, \mathbb C^4)$ (accordingly, $T_\pm C_0^\infty(\mathbb R^3, \mathbb C^4)$), and their argument relies on the statement that $C_0^\infty(\mathbb R^3, \mathbb C^4)$ is an operator core for $H$, which is only true for $V\in \overline{\mathcal P}_{\sqrt3/2}$, see Theorem~2.1.6 of \cite{BalinskyEvans}.

Trying to overcome this difficulty we have returned to the minimax principle \eqref{lambda_k}. The corresponding abstract formulation, which is the main result of our paper, naturally applies to self--adjoint operators obtained as form perturbations of symmetric sesquilinear forms.  Moreover, we only deal with the domain of the unperturbed quadratic form. In the case of Dirac operators we prove the minimax characterisation of eigenvalues \eqref{lambda_k} for all $V\in \mathcal P_1$ and a version of the Talman--Datta--Devaiah minimax principle for $V\in \overline{\mathcal P}_{2/(2/\pi+ \pi/2)}$. Our proofs are based on the ones of \cite{Gap2000} and \cite{Min1999}, but we consistently work with forms instead of operators.

The main abstract result of the article is explicitly formulated in Subsection~\ref{abstract subsection}, and the applications to Dirac operators can be found in Subsection~\ref{applications subsection}. In Section~\ref{form perturbation section} we give the definition of form perturbations, which is the key element in the construction of the operators we study. The rest of the paper contains proofs. In the appendix the (very standard) definition of the free Dirac operator is given for convenience.

Throughout the text for any sesquilinear form $f: \mathfrak Q\times \mathfrak Q\to \mathbb C$ (linear in the second argument) we say that $f$ is defined on $D[f]:= \mathfrak Q$. The corresponding quadratic form is defined on $\mathfrak Q$ by $f[x]:= f[x,x]$. If we start from a quadratic form $f$ on $D[f]$, then the corresponding sesquilinear form is naturally defined on $D[f]$ by
\[
 f[x,y]= \frac14\big(f[x+y]- f[x-y]- \text{i}f[x+ \text{i}y]+ \text{i}f[x- \text{i}y]\big).
\]
For a linear operator $A$ its domain is denoted by $D(A)$.

\subsection{The abstract minimax principle}\label{abstract subsection}

In order to treat the Dirac operators with strong Coulombic singularities, Nenciu \cite{Nen1976} has introduced the concept of \emph{form perturbations} of self--adjoint operators, which generalises the pseudo--Friedrichs extension of Kato (\cite{Kato}, VI.3.4). We will slightly modify this definition and introduce form perturbations of symmetric sesquilinear forms in Section~\ref{form perturbation section}. The concept of form perturbation is needed for the following theorem:
\begin{The}\label{B theorem}
Let a symmetric sesquilinear form $v$ be a form perturbation of a symmetric sesquilinear form $q$. Then there exists a unique self--adjoint operator $B$ satisfying the conditions 
\begin{align}
       \text{\emph{(j)}}&\quad D(B) \subset D[q]; \label{the1}\\
       \text{\emph{(jj)}}&\quad \langle Bx,y\rangle = q[x,y] + v[x,y] \quad \text{for all }\ x \in D(B), \ y\in D[q]. \label{the2}
\end{align}
Moreover,
\begin{align}
           D(B)=\bigg\{ x\in D[q]: \sup\limits_{y\in D[q]\setminus \{0\}}\frac{\big|q[x,y]+ v[x,y]\big|}{\| y \|}<\infty \bigg\}.
           \label{form6}
\end{align}
\end{The}
The proof of Theorem~\ref{B theorem} is identical to the one of Theorem~2.1 of \cite{Nen1976}.

Our main result is a minimax principle for the eigenvalues of $B$ in the gaps of its essential spectrum $\sigma_{\textrm{ess}}(B)$:

\begin{The}\label{c3}
Let a symmetric sesquilinear form $v$ be a form perturbation of a symmetric sesquilinear form $q$. Let $\mathfrak H_\pm$ be orthogonal subspaces of $\mathfrak H$ such that $\mathfrak H=\mathfrak H_{+} \oplus \mathfrak H_{-} $ and $ \Lambda_{+}$, $\Lambda_{-}$ 
the projectors onto $\mathfrak H_{+}$ and $\mathfrak H_{-}$, respectively. We assume that
\begin{align}
      \text{\emph{(i)}}&\quad \mathfrak D_{\pm}:=\Lambda_{\pm}D[q] \subset D[q];\label{the3}\\
      \text{\emph{(ii)}}& \quad  a:=\sup\limits_{x \in \mathfrak D_{-}\setminus \{ 0\}}\frac{q[x]+v[x]}{\|x\|^2}<\infty ;\label{the4}\\
      \text{\emph{(iii)}}&\quad \lambda_1 > a,\quad\text{where} \label{the6} \\
      &\quad\lambda_{k}:=\inf_{\substack{\mathfrak V \textrm{\emph{ subspace of }}\mathfrak D_{+} \\ \dim\mathfrak V=k}}\sup\limits_{x\in(\mathfrak V\oplus 
              \mathfrak D_{-})\setminus \{ 0 \}}\frac{q[x]+v[x]}{\|x\|^2}.\label{the5} 
\end{align}
Let $B$ be the self--adjoint operator defined in Theorem~\ref{B theorem} and \[b:=\inf\big(\sigma_{\mathrm{ess}}(B)\cap(a,\infty)\big)\in [a,\infty].\] For $ k\in \mathbb N$, we denote by $ \mu_k$ the
$k^{\mathrm{th}}$ eigenvalue of $B$ in the interval $(a,b)$ in non-decreasing order, counted with multiplicity, if such eigenvalue exists. If there is no $k^{\mathrm{th}}$ eigenvalue, we let $ \mu_k :=b.$

 Then
\begin{equation}
       \lambda_k = \mu_k \quad\text{for all}\quad k\in \mathbb N. \label{the7}
\end{equation}
\end{The}
The proof of Theorem~\ref{c3} can be found in Section~\ref{abstract proof section}.

\subsection{Application to Dirac operators with Coulomb singularities}\label{applications subsection}

In this subsection we elaborate and improve upon the results of \cite{Gap2000} and \cite{Min1999} using Theorem~\ref{c3}. In the following $h_0$ is the quadratic form of the free Dirac operator $H_0$ in $\mathsf L^2(\mathbb R^3, \mathbb C^4)$ with $D[h_0]= \mathsf H^{1/2}(\mathbb{R}^3, \mathbb{C}^4)$, see Appendix for more details. Let $V\in \mathcal P_1$, see \eqref{ad_pot}, and $v$ be the sesquilinear form of $V$.

It is shown in~\cite{Nen1976} that $v$ is a form perturbation of $h_0$ for $V\in \mathcal P_1$. Applying Theorem~\ref{B theorem} we define a unique self--adjoint operator $H$ in $\mathsf L^2(\mathbb{R}^3, \mathbb{C}^4)$ that satisfies
\begin{equation*}
      D(H)\subset \mathsf H^{1/2}(\mathbb{R}^3, \mathbb{C}^4)
\end{equation*}
and
\begin{equation*}
     \langle Hx,y \rangle = h_0[x,y]+v[x,y]\quad\text{for all }x\in D(H)\text{ and }y \in \mathsf H^{1/2}(\mathbb{R}^3, \mathbb{C}^4).
\end{equation*}
This construction of $H$ is by Nenciu~\cite{Nen1976} and coincides with the self--adjoint extensions constructed by Schmincke~\cite{Schmincke1972} and W\"ust~\cite{Wuest1975}.

We start with a minimax principle choosing $\Lambda_\pm$ to be the spectral projectors $P_\pm$ defined in \eqref{Ps}. 
\begin{The}\label{mainapp}
Let $h_0$, $v$ and $H$ be as defined above. Then the $k^{\mathrm{th}}$ eigenvalue $ \mu_k $ of $H$ in $(-1,1)$, counted from below with multiplicity, is given by
\begin{equation}
\mu_{k}=\inf_{\substack{\mathfrak V\text{\emph{ subspace of }}\mathfrak D_+ \\ \dim\mathfrak V=k}}\sup\limits_{x\in(\mathfrak V\oplus \mathfrak D_-)\setminus \{ 0 \}}\frac{h_0[x]+v[x]}{\|x \|^2},
\end{equation}
where $\mathfrak D_\pm:= P_\pm\mathsf H^{1/2}(\mathbb{R}^3, \mathbb{C}^4)$.
\end{The}

Another possible choice of $\Lambda_\pm$ are $T_\pm$, see \eqref{T_pm}. In this case we will have to further restrict the maximal admissible strength of the Coulomb singularity:
\begin{The}\label{mainapp2}
Let $h_0$, $v$ and $H$ be as defined above. Assume furthermore that $V\in \overline{\mathcal P}_{2/(2/\pi+ \pi/2)}$. Then the $k^{\mathrm{th}}$ eigenvalue $ \mu_k $ of $H$ in $(-1,1)$, counted from below with multiplicity, is given by
\begin{equation}
\mu_{k}=\inf_{\substack{\mathfrak V\text{\emph{ subspace of }}\mathfrak T_+ \\ \dim\mathfrak V=k}}\sup\limits_{x\in(\mathfrak V\oplus \mathfrak T_-)\setminus \{ 0 \}}\frac{h_0[x]+v[x]}{\|x \|^2},
\end{equation}
where $\mathfrak T_{\pm}:=T_\pm\mathsf H^{1/2}(\mathbb{R}^3, \mathbb{C}^4)$.
\end{The}
The proofs of Theorems~\ref{mainapp} and \ref{mainapp2} can be found in Section~\ref{Dirac proofs section}.

\section{Form perturbations}\label{form perturbation section}

In this section we define the concept of form perturbations for symmetric sesquilinear forms.

Let $q$ be a symmetric sesquilinear form on a dense domain $D[q]$ in a complex Hilbert space $\mathfrak H$.
We assume that two orthogonal projections $P_\pm$ with $P_++ P_-= \mathds 1_{\mathfrak H}$ satisfy
\begin{enumerate}[label=(\arabic*)]
      \item \label{0} $P_\pm D[q]\subset D[q]$;
      \item \label{i} $q[x_{+}]>0$ for all $x_+\in P_+D[q]\setminus\{0\}$;
      \item $q[x_-]\leqslant 0$ for all $x_-\in P_-D[q]$;
      \item \label{iii}$q[x_+, x_-]=0$ for all $x_+\in P_+D[q]$ and $x_-\in P_-D[q]$.
\end{enumerate}
For $\alpha> 0$ we define the inner product in $D[q]$ by
\begin{equation}\label{form product}
 \langle x, y\rangle_{\alpha}:= q[P_+x, P_+y]- q[P_-x, P_-y]+ \alpha\langle x, y\rangle
\end{equation}
and assume that
\begin{enumerate}[resume,label=(\arabic*)]
 \item \label{v} $\mathfrak Q_\alpha:=\big(D[q], \langle \cdot, \cdot\rangle_\alpha\big)$ is a Hilbert space (i.e., is complete).
\end{enumerate}
Note that
\begin{equation}\label{alpha equivalence}
 \|\cdot\|_\alpha^2\leqslant \|\cdot\|_{\widetilde\alpha}^2\leqslant \frac{\widetilde\alpha}\alpha\|\cdot\|_\alpha^2 \quad\text{for} \quad\widetilde\alpha> \alpha> 0,
\end{equation}
so the topology of $\mathfrak Q_\alpha$ is independent of $\alpha> 0$.
We introduce
\[U:=\mathds{1}\oplus(-\mathds{1})\quad\text{ on }P_+\mathfrak H\oplus P_-\mathfrak H.\]
\begin{enumerate}[resume,label=(\arabic*)]
 \item Let $v$ be a symmetric sesquilinear form in $\mathfrak H$ with 
 $D[v]\supseteq D[q]$.\label{form1}
 \item We assume that $v$ is bounded on $\mathfrak Q_\alpha$, i.e. there exists a constant $C_{\alpha}>0$ such that
\begin{equation}
       \big|v[x,y]\big|\leqslant C_{\alpha} \|x \|_\alpha \|y \|_\alpha \quad \text{for all }x,y\in D[q]\label{form2}.
\end{equation}
\end{enumerate}
Then $v$ defines on $\mathfrak Q_\alpha$ a bounded self--adjoint operator $ V_{\alpha} $ by 
\begin{equation}
      \langle V_{\alpha}x,y\rangle_\alpha = v[x,y]\quad\text{for all }x,y \in D[q].\label{form3}
\end{equation}
Note that by \eqref{alpha equivalence} (7) holds (or not) for all $\alpha>0$ at the same time.
\begin{enumerate}[resume,label=(\arabic*)]
 \item\label{viii} At last, we assume that for $\alpha$ big enough the operator $U+V_{\alpha}$ has a bounded inverse in $\mathfrak Q_\alpha$.
\end{enumerate}
\begin{Def}\label{form perturbation}
 If the assumptions \ref{0}--\ref{viii} are satisfied, we say that $v$ is a form perturbation of $q$.
\end{Def}
\begin{Lem}\label{spectral projections lemma}
If $q$ is a sesquilinear form of a self--adjoint operator $Q$, then the assumptions \ref{0}--\ref{iii} are satisfied if and only if \[P_+= P_Q^+:= P_Q\big((0, \infty)\big), \qquad P_-= P_Q^-:= P_Q\big((-\infty, 0]\big),\]
where $P_Q(\Omega)$ is the spectral projector of $Q$ corresponding to a Borel set $\Omega\subset \mathbb R$.
\end{Lem}
\begin{proof}
Since $q[x_{+},x_{-}]=0$ for all $x_{\pm}\in P_\pm D[q]$, $P_{+}QP_{-}=P_{-}QP_{+}=0$ holds. Hence $[P_{+},Q]=0$ and, therefore, $\big[P_{+},P_Q^\pm\big]=0$. We thus have
\[
 0\lesseqgtr q\big[P_{\pm}P_{Q}^\mp x\big]=q\big[P_{Q}^\mp P_{\pm}x\big]\lesseqgtr 0\quad\textrm{ for all }\quad x\in D[q].
\]
This implies 
\begin{equation*}
         P_{\pm}P_{Q}^\mp x= P_{Q}^\mp P_{\pm}x= 0 \quad\text{for all }x\in D[q].
\end{equation*}
Thus for every $x\in D[q]$
\[
 P_{Q}^\pm x= P_{Q}^\pm(P_+x +P_-x)= \big(1- P_{Q}^\mp\big)P_\pm x= P_\pm x.\qedhere
\]
\end{proof}
\begin{Rem}
If $q$ is a sesquilinear form of a self--adjoint operator $Q$, and $v$ is a form perturbation of $q$, then by Lemma~\ref{spectral projections lemma} $v$ is a form perturbation of $Q$ in the sense of Nenciu \cite{Nen1976}.
\end{Rem}

\section{Proof of the abstract minimax principle}\label{abstract proof section}

The inequality $\lambda_k \leqslant \mu_k$ for all $k\in\mathbb N$ follows from the proof of Theorem~1 of \cite{Min1999}. It remains to prove that $\lambda_k \geqslant \mu_k$. We follow the ideas of the proof of Theorem~1.1 of \cite{Gap2000}, but consistently work with forms instead of operators.

We first introduce a sesquilinear form
\begin{equation}\label{s}
s:= q+ v\qquad\text{on}\quad D[q]
\end{equation}
and
\begin{align}     
       s_{-}: \mathfrak D_{-}\times \mathfrak D_{-}\to \mathbb{C},\qquad s_-[x_-,y_-]:= -s[x_-,y_-].\label{dir2_1_6}
\end{align}
Furthermore, for $u> a$ let
\begin{align}   
          m_u: \mathfrak D_{-}\to [0, \infty),\qquad m_u[y_{-}]:= s_-[y_{-}]+ u\|y_-\|^2.\label{dir2_1_3}
\end{align}
By \eqref{the4}, $m_u^{1/2}$ are equaivalent norms in $\mathfrak D_{-}$.
We denote the completion of $\mathfrak D_{-}$ in $\mathfrak H$ with respect to $m_{a+1}^{1/2}$ by $\overline{\mathfrak D_-}$, and the unique continuous extensions of $m_u$
to $\overline{\mathfrak D_-}$ by $\overline{m_u}$. 
Since $s_{-}$ is continuous with respect to $m_{a+1}$, we can uniquely extend it to
\begin{equation}
      \overline{s_{-}}:\overline{\mathfrak D_-}\times \overline{\mathfrak D_-} \to \mathbb{C}.\label{dir2_1_7}
\end{equation}

For $x_{+}\in \mathfrak D_{+}$ and $u> a$ let
\begin{align}  
      \varphi_{u,x_{+}}: \mathfrak D_{-}\to \mathbb{R},\qquad \varphi_{u,x_{+}}(y_{-}):= s[x_++y_-]-u\|x_{+}+y_{-}\|^2.\label{dir2_1_5}
\end{align}
Then for $u>a$ and $x_{+}\in \mathfrak D_+$ we have
\begin{align}
           \sup\limits_{y_{-}\in \mathfrak D_-}\varphi_{u,x_{+}}(y_{-})= \sup\limits_{y_{-}\in 
           \mathfrak D_-}\big(s[x_{+}]-u\|x_{+}\|^2 + 2 \Re s[x_{+},y_{-}]-m_u[y_{-}]
           \big).\label{dir3_1_2}
\end{align}
Since the norms $m_u^{1/2}$ are equivalent to each other,
\[\sup\limits_{y_{-}\in \mathfrak D_-}\varphi_{u,x_{+}}(y_{-})<\infty\text{ for $u>a$ if and only if }x_{+}\in \mathfrak S,\]
where
\begin{align}
         \mathfrak S:=\bigg\{x_{+}\in \mathfrak D_+: \sup\limits_{y_{-}\in \mathfrak D_-\setminus
         \{0\}}\frac{\big|s[x_{+},y_{-}]\big|}{m^{1/2}_{a+1}[y_{-}]}<\infty \bigg\}\subset \mathfrak D_+.
\end{align}
For $x_{+}\in \mathfrak S$ and $u>a$, $s[x_{+},\cdot]$ extends to a linear bounded functional $s_{x_{+}}$ in the Hilbert space $(\overline{\mathfrak D_-}, \overline{m_u})$. Hence by the Riesz's theorem there exist a unique linear operator
\begin{equation}\label{dir3_1_1}
L_{u}: \mathfrak S \to\overline{\mathfrak D_-}\ \text{such that}\ s_{x_+}(y_-)= \overline{m_u}\big[L_ux_+, y_-\big]\text{ for all }y_-\in \overline{\mathfrak D_-}.
\end{equation}

Let $\overline{\varphi_{u,x_{+}}}$ be the unique continous extension of $\varphi_{u,x_{+}}$ to $\overline{\mathfrak D_-}$ for $x_{+}\in\mathfrak S$.
By \eqref{dir3_1_2} we have
\begin{equation}\label{sup simple}\begin{split}
 \sup\limits_{y_{-}\in \overline{\mathfrak D_-}}\overline{\varphi_{u,x_{+}}}(y_{-})= s[x_{+}]-u\|x_{+}\|^2 + \overline{m_u}\big[L_ux_+\big]- \inf_{y_{-}\in \overline{\mathfrak D_-}}\overline{m_u}\big[L_ux_+- y_-\big].
\end{split}\end{equation}
This obviously implies that $L_ux_+$ is the unique maximiser in \eqref{sup simple}.

\begin{Lem}
\begin{align}
      \lambda_{k}=\inf_{\substack{\mathfrak V \text{\emph{ subspace of }}\mathfrak S \\ 
       \dim\mathfrak V=k}}\sup\limits_{x\in(\mathfrak V\oplus \mathfrak D_-)\setminus \{ 0 \}}\frac{s[x]}{\|x\|^2}. \label{deflam} 
\end{align}
\end{Lem}
\begin{proof}
If for $x_{+}\in \mathfrak D_+\setminus\{ 0\}$ there exists $u\in (a, \infty)$ such that
\[
\sup_{x_{-}\in \mathfrak D_-}\frac{s[x_{+}+x_{-}]}{\|x_{+}+x_{-}\|^2}< u,
\]
then by \eqref{dir2_1_6} and \eqref{dir2_1_3}
\begin{align*}
       0>\sup\limits_{y_{-}\in \mathfrak D_-}\frac{s[x_{+}+y_{-}]-u\|x_{+}+y_{-}\|^2}{\|x_{+}+y_{-}\|^2}
       \geqslant \frac{1}{\|x_{+}\|^2}\sup\limits_{y_{-}\in \mathfrak D_-}\varphi_{u,x_{+}}(y_{-})
\end{align*}
holds. But then $x_{+}\in\mathfrak S$ and we can reformulate \eqref{the5} as \eqref{deflam}.
\end{proof}

For $u>a$ we define
\begin{align}
       g_{u}&: \mathfrak S \to\mathbb{R}, \quad g_u[x_{+}]:= \overline{\varphi_{u,x_{+}}}(L_{u}x_{+});\label{g_u}\\       
       n_{u}&: \mathfrak S \to\mathbb{R}, \quad n_u[x_{+}]:= \|x_{+}\|^2+ \|L_{u}x_{+}\|^2.\label{n_u}
\end{align}
\begin{Lem}\label{c4}
 Assume that \eqref{the3} and \eqref{the4} are satisfied. If $ a<u<u'$, then
 \begin{align}
       &\|\cdot\| \leqslant n^{1/2}_{u'} \leqslant n^{1/2}_{u} \leqslant \frac{u'-a}{u-a}n^{1/2}_{u'}; \label{Lem41}\\
       &(u'-u)n_{u'}\le g_{u}-g_{u'}\leqslant (u'-u)n_{u}. \label{Lem42}
 \end{align}
 Moreover, for any $u > a$:
 \begin{align}
       &\lambda_1 > u \quad \text{iff} \quad g_{u}[x_{+}]>0 \quad\text{for all }x_{+}\in\mathfrak S\setminus\{ 0 \};
       \label{Lem43}\\
       &\lambda_1 \geqslant u \quad \text{iff} \quad g_{u}[x_{+}]\geqslant 0 \quad\text{for all }x_{+}\in\mathfrak S.\label{Lem44}
 \end{align}
 As a consequence, \eqref{the6} is equivalent to
 \begin{align}
       \text{\emph{(iii')} For some $u>a$ , }g_{u}[x_{+}]\geqslant 0 \quad\text{for all }x_{+}\in\mathfrak S. \label{Lem45}
 \end{align}
\end{Lem}
\begin{proof}
We define (recall \eqref{dir2_1_7})
\begin{align}  
      B^\#_-: \overline{\mathfrak D_-}\to (\overline{\mathfrak D_-})^{*},\qquad (B^\#_-x_-)(y_-):=\overline{s_{-}}[x_-,y_-].\label{dir2_1_8}
\end{align}
and introduce the embedding operator
\begin{equation}\label{J}
J:\mathfrak H\to \mathfrak H^*, \qquad (Jx)(y):= \langle x,y\rangle.
\end{equation}
We first prove that
\begin{equation}\label{B+uJ}
\textrm{the operator }B^\#_-+uJ: \overline{\mathfrak D_-} \to (\overline{\mathfrak D_-})^{*}\textrm{ is invertible for all }u>a.
\end{equation}
The injectivity follows from \eqref{the4}. Now for any $ f\in (\overline{\mathfrak D_-})^{*}$ there is $c_f>0$ such that
$\big|f(y)\big|\leqslant c_f\overline{m_u}^{1/2}(y)$ for all $y \in \overline{\mathfrak D_-}$. Hence by the Riesz representation theorem there exists $x_u\in \overline{\mathfrak D_-}$ such that
\begin{align*}
       f(y)= \overline{m_u}[x_u,y]= \overline{s_{-}}[x_u,y]+u\langle x_u,y\rangle
\end{align*}
for all $y \in \overline{\mathfrak D_-}$. This implies $f(y)=\big((B^\#_-+uJ)x_u\big)(y)$ for all $y \in \overline{\mathfrak D_-}$ which 
means that $B^\#_-+uJ:\overline{\mathfrak D_-} \to (\overline{\mathfrak D_-})^{*}$ is surjective.

We know that $\overline{s_{-}}$ is a densely defined, closed and bounded below sesquilinear form in $\mathfrak H_{-}$. By the Friedrichs theorem (see e.g. \cite{Weid}, Theorem~5.37) there exists a self--adjoint operator $B_{-}$ such that
\begin{align}
             D(B_{-})&:=\big\{x \in \overline{\mathfrak D_-}:\text{ there exists }\tilde{x}\in\mathfrak H_{-}\\ &\qquad\text{ such that }\langle \tilde{x},y\rangle
             =\overline{s_{-}}[x, y]\text{ for all }y \in \overline{\mathfrak D_-} \big\},\nonumber\\
             B_{-}x &:=\tilde{x}\quad\text{for }x\in D(B_-).
\end{align}
By \eqref{dir2_1_8} we obtain
\begin{align*}
      D(B_{-})&=\big\{x \in \overline{\mathfrak D_-}:\quad B^\#_-x \in\mathfrak H_-^*\subset(\overline{\mathfrak D_-})^{*}\big\},\\
      J(B_{-}x)&=B^\#_-x\quad \text{for all }x\in D(B_{-}). 
\end{align*}
As in Lemma~2.1 of \cite{Gap2000}, using the spectral decomposition of $B_-$ we obtain for $u'>u>a$
\begin{align*}
      \big\|(B_-+u')^{-1}y\big\|\leqslant \big\|(B_-+u)^{-1}y\big\|\leqslant \frac{u'-a}{u-a}\big\|(B_-+u')^{-1}y\big\| \quad\text{for all }y\in\mathfrak H_{-}.
\end{align*}
By the density of $\mathfrak H_{-}^*$ in $(\overline{\mathfrak D_-})^{*}$ we get
\begin{align}\label{relation between B+u}
      \big\|(B^\#_-+u'J)^{-1}y\big\|\leqslant \big\|(B^\#_-+u J)^{-1}y\big\|\leqslant \frac{u'-a}{u-a}\big\|(B^\#_-+u'J)^{-1}y\big\|
\end{align}
for all $y\in (\overline{\mathfrak D_-})^{*}$. Let us introduce (recall \eqref{dir3_1_1})
\begin{align}
      \begin{aligned}
             S_+^\#:\mathfrak S\to (\overline{\mathfrak D_-})^{*}, \quad  (S_+^\#x_{+})(y_{-}):=s_{x_{+}}(y_{-}).
      \end{aligned}
\end{align}
By \eqref{dir2_1_3}, \eqref{dir2_1_8} and \eqref{B+uJ} we observe that
\begin{align*}
      L_ux_+= (B^\#_-+uJ)^{-1}S_+^\#x_{+}.
\end{align*}
Substituting this into \eqref{relation between B+u}, we obtain
\begin{align*}
      \|L_{u'}x_+\|\leqslant \|L_ux_+\|\leqslant \frac{u'-a}{u-a}\| L_{u'}x_+\| \quad\text{for all }x_+\in \mathfrak S.
\end{align*}
This together with \eqref{n_u} implies \eqref{Lem41}. The remaining statements follow in the same way as in Lemma~2.1 of \cite{Gap2000}, where the role of $F_+$ is now played by $\mathfrak S$ and we use \eqref{deflam} instead of \eqref{the5}.
\end{proof}

Let the Hilbert space $(\mathfrak X, \overline{n_u})$ be the completion of $(\mathfrak S, n_u)$. Note that by Lemma~\ref{c4} $\mathfrak X$
is contained in $\mathfrak H_+$ and does not depend on $u>a$.

By \eqref{Lem44}, $g_{u}[x_{+}]\geqslant 0$ for all $x_{+}\in\mathfrak S$ if $a<u\leqslant\lambda_1$. On the other hand, for $u\geqslant \lambda_1$, by \eqref{Lem42} and \eqref{Lem41} we obtain
\begin{align*}
        g_u \geqslant g_{\lambda_1}+(\lambda_1-u)n_{\lambda_1} \geqslant (\lambda_1-u)\Big(\frac{u-a}{\lambda_1 - a}\Big)^2 n_{u}.
\end{align*}
Hence for any $u> a$
\begin{align*}
      g_{u}\geqslant -c_u n_{u},\quad c_u := \max\bigg\{0, (u-\lambda_1)\Big(\frac{u-a}{\lambda_1 -a}\Big)^2\bigg\}.       
\end{align*}
Now we define
\begin{align}\label{h_u}
      h_{u}:\mathfrak S \longrightarrow\mathbb{R}, \qquad  h_u[x_{+}]:= g_{u}[x_{+}]+(c_u + 1)n_{u}[x_{+}].
\end{align}

We claim that $h_u^{1/2}$ and $h_{u'}^{1/2}$ are equivalent norms for $u'>u>a$. By \eqref{Lem42} and \eqref{Lem41}
\begin{align*}
       g_{u}\leqslant g_{u'}+(u'-u)n_{u} \leqslant g_{u'}+(u'-u)\Big(\frac{u'-a}{u-a}\Big)^2n_{u'},
\end{align*}
which implies that $h_{u}\leqslant g_{u'}+(1+c_{u}+u'-u)(u'-a)^2(u-a)^{-2}n_{u'}$ and that there exists a constant 
$c_2 (u,u')$ such that $h_u \leqslant c_2(u,u')h_{u'}$. By Lemma \ref{c4} we get that $g_u\geqslant g_{u'}+(u'-u)n_{u'}$ and hence
$h_{u} \geqslant g_{u'}+(u'-u+1+c_{u'})n_{u'}$ which means that there is a constant $c_1 (u,u')>0$ such that 
$h_u \geqslant c_1 (u,u')h_{u'}$. Hence the norms are equivalent.

For $u>a$ let the Hilbert space $\mathfrak G_u= (\mathfrak G, \overline{h_u})$ be the completion of $(\mathfrak{S}, h_u)$. Note that $\mathfrak G\subset \mathfrak X$ does not depend on $u$.

The extension of $g_u$ to $\mathfrak G$ is denoted by $\overline{g_u}$. It is a closed, semi-bounded quadratic form with the domain $\mathfrak G$. By the Friedrichs theorem there is a unique self--adjoint operator $T_u : D(T_u)\subset\mathfrak X \to\mathfrak X$ with the form domain $\mathfrak G$, such that $\overline{g_u}[x_{+}]= \overline{n_u}[x_{+},T_{u}x_{+}]$ for all $x_{+} \in D(T_u)$ and $\mathfrak S$ is a form-core for $T_u$.

The following lemma is a simple consequence of Courant minimax principle.
\begin{Lem}\label{c6}
Let $T$ be a self--adjoint, bounded below operator in a Hilbert space $\mathfrak X$ with the domain $D(T)$ and $t$ the corresponding sesquilinear form with the domain $D[t]$.
We define
\begin{align}
      l_k(T)&:=\inf_{\substack{\mathfrak Y \text{\emph{ subspace of} }D[t] \\ \dim\mathfrak Y=k}}\sup\limits_{x\in\mathfrak Y\setminus \{ 0 \}}
      \frac{t[x]}{\|x\|^2_{\mathfrak X}}, \label{lem61}\\
      w_k(T)&:=\mathrm{card}\big\{ k'\geqslant 1,\, l_{k'}(T)=l_{k}(T)\big\},\label{lem62}
\end{align}
and 
\begin{align}      
            T^\#: D[t]\to D[t]^{*},\qquad (T^\#z)(v):= t[z,v]\quad \text{ for all }v,z\in D[t].
\end{align}
If $ l_k(T) < \inf\sigma_{\text{\emph{ess}}}(T) $ then $ l_k(T) $ is an eigenvalue of $T$ with multiplicity $w_k(T)$. As a consequence, if $ \mathfrak{C}\subset D[t] $ is a form-core for $T$, then there is a sequence $ (\mathfrak Z_n)_{n\in \mathbb N} $ of subspaces of  $\mathfrak{C}$ with $ \dim \mathfrak Z_n=w_k(T) $ and (recall \eqref{J})
\begin{align}
      \sup_{\substack{z\in\mathfrak Z_n \\ \|z\|_{\mathfrak X}=1}}\Big\| \big(T^\#-l_{k}(T)J\big)z\Big\|_{D[t]^{*}} \to 0\qquad \text{for }n\to \infty.\label{lem63}
\end{align}
\end{Lem}
Applying Lemma~\ref{c6} we obtain
\begin{align*}
      l_k(T_u)&=\inf_{\substack{\mathfrak V \text{ subspace of }\mathfrak G \\ \dim\mathfrak V=k}}\sup\limits_{x_{+}
      \in \mathfrak V\setminus \{ 0 \}}\frac{\overline{g_u} [x_{+}]}{\overline{n_u} [x_{+}]}, \\
      w_k(T_u)&=\mathrm{card}\{k'\geqslant 1, l_{k'}(T_u)=l_{k}(T_u)\}. 
\end{align*}
If $ l_k(T_u) < \inf\sigma_{\text{ess}}(T_u) $ then $l_k(T_u)$ is an eigenvalue of $T_u$ with multiplicity $ w_k(T_u)$.
As in Lemma \ref{c6} we define for $u>a$
\begin{align}
            T_u^\#: \mathfrak G_u\to \mathfrak G_u^{*},\qquad  (T_u^\#z)(w):=\overline{g_u}[z,w]\qquad \text{ for all }w,z \in\mathfrak G.\label{dir2_6_1}
\end{align}

Starting from \eqref{deflam} and following the proof of Lemma~2.2 of \cite{Gap2000} we obtain
\begin{Lem}\label{c5}
Let \eqref{the3}, \eqref{the4} and \eqref{the6} be satisfied. Then for any $ k\geqslant 1$, $\lambda_k$ is the unique solution in $(a,\infty)$ of the non--linear equation 
\begin{equation}
       l_k(T_{\lambda})=0. \label{Lem54}
\end{equation}
\end{Lem}
We thus have $\lambda_k=\lambda_{k'}$ if and only if $l_{k'}(T_{\lambda_k})=0$. Let 
\begin{align}
         w_k:=\mathrm{card}\{k'\geqslant 1:\lambda_k=\lambda_{k'} \}. \label{dir2_6_2}
\end{align}
Then Lemma \ref{c6} implies the existence of a sequence of subspaces $(\mathfrak Z_n)_{n\in \mathbb N}$ of $\mathfrak S$ with $\dim\mathfrak Z_n =w_k$ for all
$n\geqslant 1$ such that 
\begin{align}
      \sup_{\substack{x_{+} \in\mathfrak Z_n\\ n_{\lambda_k}[x_{+}]=1}}\|T_{\lambda_k}^\#x_{+}\|_{\mathfrak G_{\lambda_k}^{*}}\to 0 \quad\text{for }n \to \infty.\label{dir2_6_3}
\end{align}

According to \eqref{the4}, \eqref{dir2_1_6} and \eqref{dir2_1_3} $(u-a)\|y_{-}\|^2\leqslant m_u[y_{-}]$ holds for all $y_{-}\in \mathfrak D_-$.
Hence by \eqref{the2} 
\begin{equation*}
             \big|s[x,y_{-}]\big|\leqslant (u-a)^{-1/2}\|Bx\| m_u^{1/2}[y_{-}]\text{ for all }y_{-}\in \mathfrak D_-,\ x\in D(B).
\end{equation*}
We thus get for $x\in D(B)$, $y_-\in \mathfrak D_-$
\begin{align*}
           \big|s[\Lambda_{+}x,y_{-}]\big|&\leqslant \big|s[x,y_{-}]\big|+\big|s[\Lambda_{-}x,y_{-}]\big|\\&\leqslant (u-a)^{-1/2}\|Bx\| m_u^{1/2}[y_{-}]+ \big|m_u[\Lambda_{-}x,y_{-}]\big|
           +|u \langle \Lambda_{-}x,y_{-}\rangle |\\
           &\leqslant \left(\frac{\|Bx\|}{\sqrt{u-a}}+m_u^{1/2}[\Lambda_{-}x]+\frac{|u|}{\sqrt{u-a}}\|\Lambda_{-}x\| \right)m_u^{1/2}[y_{-}].
\end{align*}
Hence
\begin{equation}\label{Lambda_+D(B)in S}
\Lambda_{+}D(B)\subset\mathfrak S.
\end{equation}

Let $x_+\in\mathfrak S$ and $y_+\in\Lambda_+D(B)$. Then by \eqref{dir2_6_1}
\begin{align}
       (T^\#_{\lambda_k}x_{+})(y_{+}) =g_{\lambda_k}[x_{+},y_{+}]. \label{dir2_7_1}
\end{align}
By \eqref{g_u}, \eqref{dir2_1_5} and \eqref{dir3_1_1} we have
\begin{align}
      g_{\lambda_k}[x_{+}]=s[x_{+}]-\lambda_k\|x_{+}\|^2 + 2s_{x_{+}}(L_{\lambda_k}x_{+}) -\overline{s_{-}}[L_{\lambda_k}x_{+}]- \lambda_{k}\| L_{\lambda_k}x_{+} \|^2\label{dir2_7_3}
\end{align}
for all $x_{+}\in\mathfrak S$.

For $u>a$ we define (recall \eqref{dir3_1_1})
\begin{align}\label{E_u}
            E_u: \mathfrak S\oplus \mathfrak D_- \to \overline{\mathfrak D_-},\qquad E_ux:= L_u\Lambda_{+}x - \Lambda_{-}x.
\end{align}

By \eqref{dir2_7_1} and \eqref{dir2_7_3} we obtain for all $x= x_+\oplus x_-\in\mathfrak S \oplus \mathfrak D_-$ and
$y\in D(B)$ with $y_\pm:= \Lambda_\pm y$:
\begin{equation}\begin{split}\label{long computation 1}
       &(T^\#_{\lambda_k}x_{+})(y_{+})\\
       &= s[x_{+},y_{+}]-\lambda_k \langle x_{+},y_{+}\rangle +
       2s_{x_{+}}(L_{\lambda_k}y_{+})-\overline{s_{-}}[L_{\lambda_k}x_{+},L_{\lambda_k} y_{+}]\\ &-
       \lambda_k \langle L_{\lambda_k}x_{+},L_{\lambda_k} y_{+}\rangle \\
       &= s[x,y]-s[x_{-},y]-s[x_{+},y_{-}]+2s_{x_{+}}(L_{\lambda_k}y_{+})-\overline{s_{-}}[L_{\lambda_k}x_{+},L_{\lambda_k} y_{+}]\\
       &-\lambda_k \langle x_{+}+L_{\lambda_k}x_{+},y_{+}+L_{\lambda_k} y_{+}\rangle
\end{split}\end{equation}
For all $x_{-}\in \mathfrak D_-$ and $y\in D(B)$ such that $\Lambda_{\pm}y=y_{\pm}$ we get
\[
s_{y_{+}}(x_{-})=s[y_{+},x_{-}]=s[y,x_{-}]-s[y_{-},x_{-}]=\langle By,x_{-}\rangle + \overline{s_{-}}[y_{-},x_{-}].
\]
Thus for all $x_{-}\in \overline{\mathfrak D_-}$ and $y\in D(B)$ such that $\Lambda_{\pm}y=y_{\pm}$
\begin{align}
        s_{y_{+}}(x_{-})=\langle By,x_{-}\rangle + \overline{s_{-}}[y_{-},x_{-}]\label{dir2_7_2}
\end{align}
holds. Now by \eqref{dir3_1_1}
\[
 s_{x_+}(L_{\lambda_k}y_+)= \overline{m_{\lambda_k}}[L_{\lambda_k}x_+, L_{\lambda_k}y_+]= \overline{s_{y_+}(L_{\lambda_k}x_+)}.
\]
This together with \eqref{dir2_7_2} implies
\begin{equation}\label{long computaion 2}\begin{split}
 &2s_{x_{+}}(L_{\lambda_k}y_{+})- s[x_{+},y_{-}]= \overline{s_{y_+}(L_{\lambda_k}x_+)}+ \big(s_{x_{+}}(L_{\lambda_k}y_{+})- s_{x_+}(y_-)\big)\\ &=\overline{s_{y_+}(L_{\lambda_k}x_+)}+ s_{x_+}(L_{\lambda_k}y_+- y_-)= \langle y_{-},Bx\rangle + \overline{s_{-}}[y_{-},x_{-}]+ s_{x_+}(E_{\lambda_k}y),
\end{split}\end{equation}
Inserting \eqref{long computaion 2} into \eqref{long computation 1} and using \eqref{E_u}, we obtain
\begin{align*}
       &(T^\#_{\lambda_k}x_{+})(y_{+})=\langle x + E_{\lambda_k} x,(B-\lambda_k )y\rangle\\
       &+s_{x_{+}}(E_{\lambda_k}y)-\overline{s_{-}}(L_{\lambda_k}x_{+},E_{\lambda_k}y)-\lambda_k \langle L_{\lambda_k}x_{+},E_{\lambda_k}y\rangle.
\end{align*}
By \eqref{dir3_1_1} and \eqref{dir2_1_3} all the terms in the last line cancel. We thus get that for $x_{+}\in\mathfrak S$, $y_{+}\in \Lambda_{+}D(B)$ the relation
\begin{align}
       (T^\#_{\lambda_k}x_{+})(y_{+})=\langle x_{+} +L_{\lambda_k}x_{+},(B-\lambda_k)y\rangle\label{dir2_7_4}
\end{align}
holds for any $y\in D(B)$ with $\Lambda_{+}y=y_{+}$.

We now estimate $h_{\lambda_{k}}$. Let $y\in D(B)$. By \eqref{Lambda_+D(B)in S} and \eqref{h_u} we get
 \begin{align}\label{h via n and g}
       h_{\lambda_k}[\Lambda_{+}y]=(c_{\lambda_k}+1)n_{\lambda_k}[\Lambda_{+}y]+g_{\lambda_k}[\Lambda_{+}y].
 \end{align}
Now by \eqref{dir2_6_1} and\eqref{dir2_7_4}
\begin{equation}\label{dir2_8_1}\begin{split}
      \big|g_{\lambda_k}[\Lambda_{+}y]\big|&=\big|(T^\#_{\lambda_{k}}\Lambda_{+}y)(\Lambda_{+}y)\big|= \big|\langle (1+L_{\lambda_k})\Lambda_{+}y,(B-\lambda_k)y\rangle\big|\\
      &=\big|\langle y+E_{\lambda_k}y,(B-\lambda_k)y\rangle\big|\leqslant \| y+E_{\lambda_k}y \|\big\|(B-\lambda_k)y\big\|\\ &\leqslant\big(1+|\lambda_k|\big)\| y+E_{\lambda_k}y \|\| y \|_{D(B)}. 
\end{split}\end{equation}
Here
\[
 \| y \|_{D(B)}:= \big(\|By\|^2+ \|y\|^2\big)^{1/2}
\]
is the graph norm on $D(B)$.

By \eqref{dir2_7_2}, \eqref{dir3_1_1} and \eqref{dir2_1_3} we obtain
\begin{align*}
         \langle By,E_{\lambda_{k}}y\rangle+\overline{s_{-}}[\Lambda_{-}y,E_{\lambda_{k}}y]=\overline{s_{-}}
         [L_{\lambda_{k}}\Lambda_{+}y,E_{\lambda_{k}}y]+\lambda_{k}\langle L_{\lambda_{k}}\Lambda_{+}y,E_{\lambda_{k}} y \rangle,
\end{align*}
which by \eqref{E_u}, \eqref{dir2_1_6} and \eqref{the4} implies 
\begin{align}\label{estimate on E}
       \big\| (B-\lambda_{k}) y \big\|\| E_{\lambda_{k}}y \|\geqslant \big|\langle (B-\lambda_{k})y,E_{\lambda_{k}}y\rangle\big|\geqslant (\lambda_{k}-a)\|E_{\lambda_{k}} y \|^2.
\end{align}
Substituting \eqref{estimate on E} into \eqref{dir2_8_1}, we obtain
\begin{align}\label{g estimate}
      \big|g_{\lambda_k}[\Lambda_{+}y]\big|\leqslant \big(1+|\lambda_k |\big)\bigg(1+\frac{1+|\lambda_k |}{\lambda_k -a}\bigg) \| y \|_{D(B)}^2.
\end{align}
By \eqref{n_u}, \eqref{E_u} and \eqref{estimate on E},
\begin{align}\label{estimate on n}
       n_{\lambda_k}[\Lambda_{+}y]= \|\Lambda_{+}y +L_{\lambda_k}\Lambda_{+}y\|^2= \| y + E_{\lambda_k}y\|^2\leqslant \bigg(1+\frac{1+|\lambda_k |}{\lambda_k -a}\bigg)^2 \| y \|_{D(B)}^2.
\end{align}
Substituting \eqref{g estimate} and \eqref{estimate on n} into \eqref{h via n and g} we find a constant $c(\lambda_k , a)>0$ such that
\begin{align}
       h^{1/2}_{\lambda_k}[\Lambda_{+}y]\leqslant c(\lambda_k , a) \| y \|_{D(B)}\quad\text{for all }y \in D(B).\label{dir2_8_8}
\end{align}
By \eqref{dir2_7_4} and \eqref{dir2_8_8} we get for all $x_{+}\in \mathfrak S$, $y_{+}\in \big(\Lambda_{+}D(B)\big)\setminus\{ 0 \}$ and $y\in D(B)$ such that $\Lambda_{+}y=y_{+}$:
\begin{align}
       \frac{\big|(T^\#_{\lambda_k}x_{+})(y_{+})\big|}{h^{1/2}_{\lambda_k }(y_{+})}\geqslant \frac{\big|\langle x_{+} +L_{\lambda_{k}}x_{+}, (B-\lambda_{k})y\rangle \big|}{c(\lambda_k , a) \| y \|_{D(B)}}.\label{dir2_8_9}
\end{align}
According to \eqref{Lambda_+D(B)in S}, for $x_+\in \mathfrak G_{\lambda_k}$
\begin{align*}
       \|T^\#_{\lambda_{k}}x_{+} \|_{\mathfrak G_{\lambda_k}^{*}} \geqslant\sup_{y_{+}\in \left(\Lambda_{+}D(B)\right)\setminus \{0 \}}
       \frac{\big|(T^\#_{\lambda_k}x_{+})(y_{+})\big|}{h^{1/2}_{\lambda_k}[y_{+}]}.
\end{align*}
From this and \eqref{dir2_6_3} it follows that
\begin{align*}
         \sup\limits_{\substack{x_{+}\in\mathfrak Z_n\\ n_{\lambda_k}[x_{+}]=1}}\sup\limits_{y_{+}\in \left(\Lambda_{+}D(B)\right)\setminus \{0 
         \}}\frac{\big|(T^\#_{\lambda_k}x_{+})(y_{+})\big|}{h^{1/2}_{\lambda_k}[y_{+}]}\to 0\quad
       \text{for }n\to\infty.
\end{align*}
Hence we get by \eqref{dir2_8_9}
\begin{align}
         \sup\limits_{\substack{x_{+}\in\mathfrak Z_n\\ n_{\lambda_k}[x_{+}]=1}}\sup\limits_{y\in D(B)\setminus \{0 \}}
         \frac{\big|\langle x_{+} +L_{\lambda_{k}}x_{+},(B-\lambda_{k})y\rangle \big|}{\| y \|_{D(B)}}\to 0\quad\text{for }n\to\infty.
         \label{dir2_8_10}
\end{align}

We now prove that either $ \lambda_k \in \sigma_{\text{ess}}(B)\cap (a,\infty) $ or $\lambda_k$ is an 
eigenvalue of $B$ in $ (a,\infty) $ with multiplicity greater than or equal to $w_k$. 
First we define
\begin{align}      
      \tilde{\mathfrak Z}_n := (1+L_{\lambda_k})\mathfrak Z_n.
\end{align}
We know that $\dim \mathfrak Z_n = w_k$ and so 
\begin{align}
      \dim\tilde{\mathfrak Z}_n = w_k \label{dir2_9_1}.
\end{align}
Relations \eqref{dir2_8_10} and \eqref{dir2_9_1} imply the existence of sequences $ (\tilde{x}_n^{(l)})_{n\in \mathbb N} \subset\mathfrak H$, $l \in \{ 1,2,...,w_k\}$ such that $\{\tilde{x}_{n}^{(1)},...,\tilde{x}_{n}^{(w_k)}\}$ is orthonormal in $\mathfrak H$
for all $n$ and
\begin{align*}
      \lim\limits_{n\to\infty}
      \sup\limits_{y\in D(B)\setminus\{ 0 \}}
      \frac{\big|\langle \tilde{x}_n^{(l)},(B-\lambda_k)y\rangle\big|}{ \| y \|_{D(B)}}
      \to 0\quad\text{for all }l \in \{ 1,2,...,w_k\}. 
\end{align*}
Since $D(B)$ is dense in $\mathfrak H$ with respect to $\|\cdot\|$, without loss of generality $ (\tilde{x}_n^{(l)}) \subset D(B)$ for all $l \in \{ 1,2,...,w_k\}$. Since $ \|y\|_{D(B)}=\big\| (B+\mathrm{i})y\big\| $ and $\mathfrak H=(B+\mathrm{i})D(B)$, we  conclude that
\begin{align*}
      \lim_{n\to\infty}\sup_{\substack{y\in\mathfrak H\setminus\{ 0 \}\\ \| y \| =1}}
      \big|\langle (B-\mathrm{i})^{-1}(B-\lambda_k)\tilde{x}_n^{(l)},y\rangle\big|
      \to 0 \quad\text{ for all }l \in \{ 1,2,...,w_k\},
\end{align*}
which means that $\lim\limits_{n \to \infty}(B-\mathrm{i})^{-1}(B-\lambda_k)\tilde{x}_n^{(l)}=0$
for all $l \in \{ 1,2,...,w_k \}$. If $\lambda_k\notin \sigma (B)$ then $(B-\lambda_k)^{-1}(B-\mathrm{i})$ is a bounded operator and thus
for $ l \in \{ 1,2,...,w_k\} $ we get 
\begin{align*}
      1=\|\tilde{x}_{n}^{(l)}\| \leqslant\big\|(B-\lambda_k)^{-1}(B-\mathrm{i})\big\|\big\|(B-\mathrm{i})^{-1}(B-\lambda_k)\tilde{x}_n^{(l)}\big\|\to 0\
      \text{for }n\to\infty,
\end{align*} 
which is a contradiction. Hence either $ \lambda_k \in \sigma_{\text{ess}}(B)\cap (a,\infty) $ or $ \lambda_k \in (a,\infty )$ is an eigenvalue of $B$ with multiplicity greater than or equal to $ w_k$. This implies that $\lambda_k\geqslant\mu_k$ for all $k\in \{1,w_1\}$. By induction we conclude that $\lambda_k\geqslant\mu_k$ for all $k\geqslant 1$.

\section{Applications to Dirac operators with singular potentials: proofs}\label{Dirac proofs section}
\subsection{Proof of Theorem~\ref{mainapp}}

We want to apply Theorem \ref{c3} with $q:=h_0$. The assumption (i) obviously holds; the assumptions (ii) with $a= -1$ follows from the non--positivity of $V$.
It remains to prove (iii).

By monotonicity and \eqref{ad_pot} it is clearly enough to deal with the case
\[
V(x)= V_{\tilde\nu, 0}(x):= -\frac{\tilde\nu}{|x|}\mathds1_{\mathbb{C}^4}.
\]
For this we consider $V_{\tilde\nu, 0}$ as an element of a family of potentials
\[
 V_{\nu, \varepsilon}(x):=\frac{-\nu}{|x|+\varepsilon}\mathds1_{\mathbb{C}^4}, \quad \nu\in[0,\tilde\nu], \quad \varepsilon \in [0,\infty).
\]
In the First Step of the proof of Theorem~4.2 in \cite{Gap2000} it is proved that for $\varepsilon> 0$ the first minimax value $\lambda_{1}(V_{\nu,\varepsilon})$ of $H_0+ V_{\nu, \varepsilon}$ satisfies
\begin{align}\label{french positivity}
            \lambda_{1}(V_{\nu,\varepsilon})\geqslant 0\text{ for all }\nu \in [0,\tilde\nu]\text{ and }\varepsilon > 0.
\end{align}
For $\nu\in [0,\tilde\nu]$ and $\varepsilon \in [0,\infty)$ we define (cf. \eqref{s})
\begin{align}
         s_{\nu,\varepsilon}: h_{0}+ v_{\nu, \varepsilon}\qquad\text{on}\quad D[s_{\nu, \varepsilon}]:=\mathsf H^{1/2}(\mathbb{R}^3, \mathbb{C}^4),
\end{align}
where $v_{\nu,\varepsilon}$ is the sesquilinear form of $V_{\nu,\varepsilon}$, and 
\begin{align}\label{g_nu,eps}
         g_{\nu,\varepsilon}:\mathfrak D_{+}\to \mathbb{R}\cup\{\infty\},\qquad g_{\nu,\varepsilon}[x_{+}]:= \sup\limits_{x_{-}\in \mathfrak D_{-}}s_{\nu,\varepsilon}[x_{+}+x_{-}].
\end{align}
Introducing
\begin{equation}\label{m}
 m_{\nu, \varepsilon}:\mathfrak D_-\to [0, \infty), \qquad m_{\nu, \varepsilon}[x_-]:= -s_{\nu,\varepsilon}[x_-]
\end{equation}
we observe that $\mathfrak D_-$ is closed with respect to the norm $m_{\nu, \varepsilon}^{1/2}$, which is equivalent to the $\mathsf H^{1/2}$--norm on $\mathfrak D_-$.
As in the proof of Theorem \ref{c3} there exists a linear operator $L_{\nu,\varepsilon}: \mathfrak S\to\mathfrak D_-$ such that
\begin{align}\label{L_nu,eps}
       g_{\nu,\varepsilon}[x_{+}]=s_{\nu,\varepsilon}[x_{+}+L_{\nu,\varepsilon}x_+].
\end{align}
By the equivalence of $m_{\nu, \varepsilon}^{1/2}$ and $\mathsf H^{1/2}$--norm on $\mathfrak D_-$ we observe that $\mathfrak S=\mathfrak{D}_{+}$.
Letting $x_{-}^{*}:= L_{\nu,\varepsilon}x_+$ and using that $x_{-}^{*}$ is a maximizer of $s_{\nu,\varepsilon}[x_{+}+\cdot]$, we obtain
\begin{align*}
         0=\frac{\text d}{\text d\alpha}\Big|_{\alpha =0}s_{\nu, \varepsilon}[x_{+}+x_{-}^{*}+\alpha y_{-}]\text{ for all }y_{-}\in \mathfrak D_{-}.
\end{align*}    
Let us now assume that
\begin{equation}\label{x_+}
x_{+}\in \mathfrak C_+:= P_+C_{0}^{\infty}(\mathbb{R}^3;\mathbb{C}^4)\subset H^{1}(\mathbb{R}^3;\mathbb{C}^4).
\end{equation}
Then
\begin{align}
         \langle P_{-}V_{\nu, \varepsilon}x_{+},y_{-}\rangle= -h_0[x_{-}^{*},y_{-}]-v_{\nu, \varepsilon}[x_{-}^{*},y_{-}]= m_{\nu, \varepsilon}[x_{-}^{*},y_{-}]\quad\text{for all }y_{-}\in \mathfrak D_{-}.\label{eom}
\end{align}
We observe that $c_{\nu, \varepsilon}:=-h_0-v_{\nu, \varepsilon}$ defined on $\mathfrak D_{-}$ is a densely defined, closed, symmetric and  bounded below sesquilinear form in $\mathfrak H_{-}:= P_-\mathsf L^2(\mathbb{R}^3, \mathbb{C}^4)$. By Friedrichs theorem there is a unique self--adjoint operator $C_{\nu, \varepsilon}$ in $\mathfrak H_{-}$ corresponding to $c_{\nu, \varepsilon}$.
Moreover,  for all $\nu \in [0,\tilde\nu]$ and $\varepsilon \in [0,\infty)$ we have
\begin{align}
             C_{\nu, \varepsilon}&\geqslant\mathds1_{\mathfrak H_-}\label{prop1}
\end{align}
and
\begin{align}
             D(C_{\nu, 0})&\subset D(C_{\nu, \varepsilon}).
             \label{prop2}
\end{align}
Relation \eqref{eom} implies that
$ x_{-}^{*} \in D(C_{\nu, \varepsilon})$ and
\begin{align*}
x_{-}^{*}= L_{\nu, \varepsilon}x_+=C_{\nu, \varepsilon}^{-1}P_{-}V_{\nu, \varepsilon}x_{+}, \quad\text{for all }x_+\in \mathfrak C_+.
\end{align*}
Now we claim that 
\begin{align}
\|C_{\nu, 0}^{-1}P_{-}V_{\nu, 0}x_{+}-C_{\nu, \varepsilon}^{-1}P_{-}V_{\nu, \varepsilon}x_{+}\|\to 0\quad \text{for }\varepsilon \searrow 0. \label{maxcon}
\end{align} 
With the help of the triangle inequality and the resolvent identity (which we can apply by \eqref{prop2}) we can estimate
\begin{equation}\begin{split}\label{triangle}
      &\|C_{\nu, 0}^{-1}P_{-}V_{\nu, 0}x_{+}- C_{\nu, \varepsilon}^{-1}P_{-}V_{\nu, \varepsilon}x_{+}\|\\ &\leqslant \big\|C_{\nu, \varepsilon}^{-1}(C_{\nu, \varepsilon}-C_{\nu, 0})C_{\nu, 0}^{-1}P_{-}V_{\nu, 0}x_{+}\big\|+ \big\|C_{\nu, \varepsilon}^{-1}P_{-}(V_{\nu, 0}- V_{\nu, \varepsilon})x_{+}\big\|.
\end{split}\end{equation}
The last term tends to zero as $\varepsilon \searrow 0$ by \eqref{prop1} and dominated convergence (Note that $V_{\nu, 0}x_+\in \mathsf L^2(\mathbb{R}^3, \mathbb{C}^4)$ by the Hardy inequality and \eqref{x_+}).

Let $y_-:=C_{\nu, 0}^{-1}P_{-}V_{\nu, 0}x_{+}$. Since $y_{-}\in D(C_{\nu, 0})$ and 
$C_{\nu, 0}\geqslant -V_{\nu, 0}\geqslant 0$ we get $y_{-}\in D(V_{\nu, 0})$. Hence, again by dominated convergence,
\begin{align*}
       \big\|(C_{\nu, \varepsilon}-C_{\nu, 0})y_{-}\big\|= \big\|(V_{\nu, 0} -V_{\nu, \varepsilon})y_{-}\big\|\to 0\text{ for }\varepsilon \searrow 0.
\end{align*}
The claim \eqref{maxcon} is thus proven.

By \eqref{eom} we have
\begin{equation}\begin{split}\label{abc}
      g_{\nu, \varepsilon}[x_{+}]&=s_{\nu, \varepsilon}[x_{+}+C_{\nu, \varepsilon}^{-1}P_{-}V_{\nu, \varepsilon}x_{+}]\\
      &=s_{\nu, \varepsilon}[x_{+}]+ s_{\nu, \varepsilon}[C_{\nu, \varepsilon}^{-1}P_{-}V_{\nu, \varepsilon}x_{+},x_{+}]\\ &= s_{\nu, \varepsilon}[x_{+}]+ \langle C_{\nu, \varepsilon}^{-1}P_{-}V_{\nu, \varepsilon}x_{+},V_{\nu, \varepsilon}x_{+}\rangle .
\end{split}\end{equation}
By \eqref{maxcon}, \eqref{abc} and dominated convergence we get
\begin{align}
         g_{\nu, 0}[x_{+}]=\lim\limits_{\varepsilon\searrow 0} g_{\nu, \varepsilon}[x_{+}],\quad\text{for all }x_{+}\in \mathfrak C_+\text{ and }\nu \in [0, \tilde\nu].  \label{conQ}
\end{align}
Let $x_{+}\in \mathfrak D_{+}$, $\nu \in [0, \tilde\nu]$ and $\varepsilon \in [0, \infty)$ be arbitrary.
By \eqref{L_nu,eps} we obtain
\begin{equation}\label{qq}
 g_{\nu, \varepsilon}[x_+]= s_{\nu, \varepsilon}[x_+]+ s_{\nu, \varepsilon}[L_{\nu, \varepsilon}x_+]+ 2\Re s_{\nu, \varepsilon}[x_+, L_{\nu, \varepsilon}x_+].
\end{equation}
Setting $y_-:= x_-^*= L_{\nu, \varepsilon}x_+$ in \eqref{eom} we can rewrite the last term in \eqref{qq}:
\begin{equation}\label{2Re}
 2\Re s_{\nu, \varepsilon}[x_+, L_{\nu, \varepsilon}x_+]= 2v_{\nu, \varepsilon}[x_{+},L_{\nu, \varepsilon}x_+]= 2m_{\nu, \varepsilon}[L_{\nu, \varepsilon}x_+].
\end{equation}
Combining \eqref{2Re}, \eqref{qq} and \eqref{m} we arrive at
\begin{equation}\label{new g}
 g_{\nu, \varepsilon}[x_+]= s_{\nu, \varepsilon}[x_+]+ m_{\nu, \varepsilon}[L_{\nu, \varepsilon}x_+].
\end{equation}
Now by \eqref{2Re}, the Kato inequality and \eqref{m}
\begin{equation}\label{Kato}\begin{split}
 &m_{\nu, \varepsilon}[L_{\nu, \varepsilon}x_+]= v_{\nu, \varepsilon}[x_{+},L_{\nu, \varepsilon}x_+] \leqslant \big(-v_{\nu, \varepsilon}[x_+]\big)^{1/2}\big(-v_{\nu, \varepsilon}[L_{\nu, \varepsilon}x_+]\big)^{1/2}\\ & \leqslant \sqrt{\frac\pi2}\|x_+\|_{\mathsf H^{1/2}(\mathbb{R}^3, \mathbb C^4)}\big(m_{\nu, \varepsilon}[L_{\nu, \varepsilon}x_+]\big)^{1/2},
\end{split}\end{equation}
i.e.
\begin{equation}\label{true Kato}
 m_{\nu, \varepsilon}[L_{\nu, \varepsilon}x_+]\leqslant \frac\pi2\|x_+\|^2_{\mathsf H^{1/2}(\mathbb{R}^3, \mathbb C^4)}.
\end{equation}
This shows that the right hand side of \eqref{new g} is continuous in $x_+$ in the $\mathsf H^{1/2}$--norm. Thus by density the non--negativity of $g_{\nu, \varepsilon}$ on $\mathfrak D_{+}$ is equivalent to its non--negativity on $\mathfrak C_+$ for all $\varepsilon \in [0,\infty)$ and $\nu \in [0, \tilde \nu]$.
For $\varepsilon> 0$ and $\nu\in [0, \tilde \nu]$, \eqref{Lem45} and \eqref{french positivity} imply $g_{\nu, \varepsilon}[x_{+}]\geqslant 0$ for all $x_{+}\in \mathfrak C_{+}$. According to \eqref{conQ}, $g_{\nu, \varepsilon}[x_{+}]\geqslant 0$ also holds for $\varepsilon= 0$ for all $x_{+}\in \mathfrak C_+$, and thus for all $x_{+}\in \mathfrak D_+$. Another application of \eqref{Lem45} finally yields $\lambda_1\geqslant 0$.

\subsection{Proof of Theorem~\ref{mainapp2}}
The statement follows from Theorem \ref{c3} with $q:=h_0$. The assumption (i) obviously holds; the assumption (ii) with $a= -1$ follows from the non--positivity of $V$.
To establish (iii) we observe that for any 2--spinor\\ $\varphi\in\mathsf H^{1/2}(\mathbb R^3, \mathbb C^2)$ the 4--spinor
\[
 \begin{pmatrix}\varphi\\ \\ \displaystyle\mathcal F^{-1}\frac{\boldsymbol{\sigma}\cdot\mathbf{p}}{p}\sqrt{\frac{\sqrt{p^2+ 1}-1}{\sqrt{p^2+ 1}+1}}\mathcal F\varphi\end{pmatrix}
\]
(where $\mathcal F$ is the Fourier transform) belongs to $P_{H_0}\big([1, \infty)\big)\mathsf H^{1/2}(\mathbb R^3, \mathbb C^4)$, which follows from e.g. Subsection 1.4.2 of \cite{Thaller}. Hence (iii) is an easy consequence of Theorem~1 of \cite{Tix}.

\section*{Appendix: free Dirac operator\footnote{See e.g. \cite{Thaller}, Chapter I}}

In $\mathsf L^{2}(\mathbb{R}^3,\mathbb{C}^4)$ the free Dirac operator
\[
H_{0}= -\mathrm{i}\boldsymbol{\alpha} \cdot \nabla + \beta
\]
is self--adjoint on the domain $D(H_0)= \mathsf H^{1}(\mathbb{R}^3,\mathbb{C}^4)$. Here
       $\boldsymbol{\alpha}= (\alpha_{1}, \alpha_{2}, \alpha_{3})$ and $\beta$ are defined as
       \begin{align*}
\beta:=\begin{pmatrix}
 \mathds{1}_{\mathbb{C}^2} & 0 \\
0 & \mathds{1}_{\mathbb{C}^2} 
\end{pmatrix},
\qquad
\alpha_k:=\begin{pmatrix}
 0 & \sigma_k \\
\sigma_k & 0 
\end{pmatrix},
\qquad k=1,2,3;
       \end{align*}
where $\sigma_k$ are the Pauli matrices:
\begin{align*}
\sigma_1:=\begin{pmatrix}
 0 & 1 \\
1 & 0
\end{pmatrix},
\qquad \sigma_2:=\begin{pmatrix}
 0 & -\mathrm{i} \\
\mathrm{i} & 0 
\end{pmatrix},
\qquad 
\sigma_3:=\begin{pmatrix}
 1 & 0 \\
0 & -1
\end{pmatrix}.
\end{align*}

\paragraph{Acknowledgement:} The authors were partially supported by the DFG through SFB-TR 12.

\bibliographystyle{plain}
\bibliography{lit}

\end{document}